\documentclass[journal]{IEEEtran} 

\usepackage{graphicx}
\usepackage{booktabs}
\usepackage{amssymb}
\usepackage{amsthm,algorithm,algorithmic}
\usepackage{blkarray}
\usepackage{multirow}
\usepackage{romannum}
\usepackage{amsfonts}
\usepackage[cmex10]{amsmath}
\usepackage{subcaption}
\usepackage{lipsum}
\usepackage{xcolor}
\usepackage{pdfpages}
\usepackage{breqn}
\usepackage{bm}
\usepackage{comment}
\usepackage{float}
\usepackage{thmtools}
\usepackage{mathtools}
\usepackage[scr]{rsfso}
\usepackage[noadjust]{cite}

\declaretheoremstyle[%
  headfont=\bfseries,%
  headpunct={},%
  postheadspace=\newline,%
  notefont=\normalfont\bfseries,%
  notebraces={--~}{},
]{definitionstyle}

\usepackage[noadjust]{cite}

\newtheorem{lemma}{Lemma}

\makeatletter
\let\old@ps@headings\ps@headings
\let\old@ps@IEEEtitlepagestyle\ps@IEEEtitlepagestyle
\def\confheader#1{%
	\def\ps@headings{%
		\old@ps@headings%
		\def\@oddhead{\strut\hfill#1\hfill\strut}%
		\def\@evenhead{\strut\hfill#1\hfill\strut}%
	}%
	\def\ps@IEEEtitlepagestyle{%
		\old@ps@IEEEtitlepagestyle%
		\def\@oddhead{\strut\hfill#1\hfill\strut}%
		\def\@evenhead{\strut\hfill#1\hfill\strut}%
	}%
	\ps@headings%
}
\makeatother

\begin{document}
\bstctlcite{IEEEexample:BSTcontrol}
\pagenumbering{arabic}

\title{Queuing Analysis of Opportunistic Cognitive Radio IoT Network with Imperfect Sensing}

\author{Asif Ahmed Sardar, Dibbendu Roy, Washim Uddin Mondal and Goutam Das}

\maketitle

\begin{abstract}
In this paper, we analyze a Cognitive Radio-based Internet-of-Things (CR-IoT) network comprising a Primary Network Provider (PNP) and an IoT operator. The PNP uses its licensed spectrum to serve its users. The IoT operator identifies the white-space in the licensed band at regular intervals and opportunistically exploits them to serve the IoT nodes under its coverage. IoT nodes are battery-operated devices that require periodical energy replenishment. We employ the Microwave Power Transfer (MPT) technique for its superior energy transfer efficiency over long-distance. The white-space detection process is not always perfect and the IoT operator may jeopardize the PNP's transmissions due to misdetection. To reduce the possibility of such interferences, some of the spectrum holes must remain unutilized, even when the IoT nodes have data to transmit. The IoT operator needs to decide what percentage of the white-space to keep unutilized and how to judiciously use the rest for data transmission and energy-replenishment to maintain an optimal balance between the average interference inflicted on PNP's users and the Quality-of-Service (QoS) experienced by IoT nodes. Due to the periodic nature of the spectrum-sensing process, Discrete Time Markov Chain (DTMC) method can realistically model this framework. In literature, activities of the PNP and IoT operator are assumed to be mutually exclusive, for ease of analysis. Our model incorporates possible overlaps between these activities, making the analysis more realistic. Using our model, the sustainability region of the CR-IoT network can be obtained. The accuracy of our analysis is demonstrated via extensive simulation.
\end{abstract}

\IEEEpeerreviewmaketitle
	
\section{Introduction}
In the near future, 5G technology is expected to revolutionize wireless networks. One of the key components of the 5G network is envisioned to be the Internet-of-Things (IoT) \cite{zhang20196g}. IoT is a network of interconnected devices that are uniquely addressable, based on standard communication protocols \cite{perera2013context}. Though the devices may form IoT networks via either wired or wireless communication technology, the latter is more suitable in terms of cost-effectiveness and providing a connection to remote users.

\textit{Smart City} is going to emerge as one of the most important applications of IoT in the foreseeable future. Major goals of smart city are - efficient use of public resources, provide a better Quality-of-Service (QoS) to the citizens and reduce operational costs of public administrations. In smart city applications, data is expected to be collected by the IoT sensors from different environmental sources and sent to a central processor for future decisions \cite{samir2019uav}. For the majority of these applications (such as - traffic congestion tracking, air quality, and noise monitoring, etc.), the delay requirement is not very stringent and can withstand up to few minutes \cite{zanella2014internet}. So, it is not economical to use dedicated bandwidth for such services, as that can lead to underutilization of spectrum, which is already a scarce resource. The idea of incorporating Cognitive Radio Network (CRN) into IoT has been explored in recent works \cite{tragos2013cognitive, salameh2019spectrum, gu2019minimizing, ansere2019reliable}. Cognitive Radio enabled IoT network is generally referred to as a CR-IoT network. A CR-IoT system can operate either in \textit{Opportunistic} mode (using licensed spectrum during owner's inactivity) or in \textit{Underlay} mode (using licensed spectrum any time while maintaining a tolerable interference on the owner).

Two major practical constraints of a CR-IoT system are low battery life and low data storage capacity of the IoT nodes. Data needs to be transferred to a central system with larger data storage and superior processing capability to prevent data loss due to the limited buffer size of the IoT sensors. The other primary concern of a CR-IoT system is maintaining the energy requirements of the devices connected to it. Energy-harvesting is an effective solution for prolonging a CR-IoT network's lifespan by providing energy to its IoT nodes and keeping them operational \cite{kamalinejad2015wireless}. In our work, we employ Wireless Power Transfer (WPT) method via microwave frequency as a means for the nodes to harvest energy, which is also known as the Microwave Power Transfer (MPT) \cite{bi2015wireless}. The IoT nodes can utilize the spectrum holes for energy harvesting and data transfer in a time-switching fashion. Therefore, it is of utmost importance to inspect the proper balance between these two actions so that both the QoS requirements of the IoT operator and the interference limit imposed by the PNP are satisfied.

In literature, few works are available that analyze the relationship between different QoS metrics and energy harvesting \cite{mohapatra2018spectrum, aslam2018energy, salameh2018spectrum}. All of these works employ energy harvesting from ambient Radio Frequency (RF) sources (for example - spectrum owner's transmission, environmental source). These traditional methods of energy harvesting are not only energy inefficient but also are unable to supply energy to the IoT sensors over a large distance. The MPT method has emerged as the alternate energy harvesting method that can provide energy to the IoT nodes efficiently over a large distance via beamforming \cite{huang2015cutting}. Unfortunately, one of the major drawbacks of MPT is it can interfere with the license owner's transmission. So, MPT should be used in a CR-IoT network that operates in opportunistic mode. To the best of our understandings, this is the first work that incorporates the MPT technique in a CR-IoT system.

\begin{figure}
	\centering
	\includegraphics[width=87 mm]{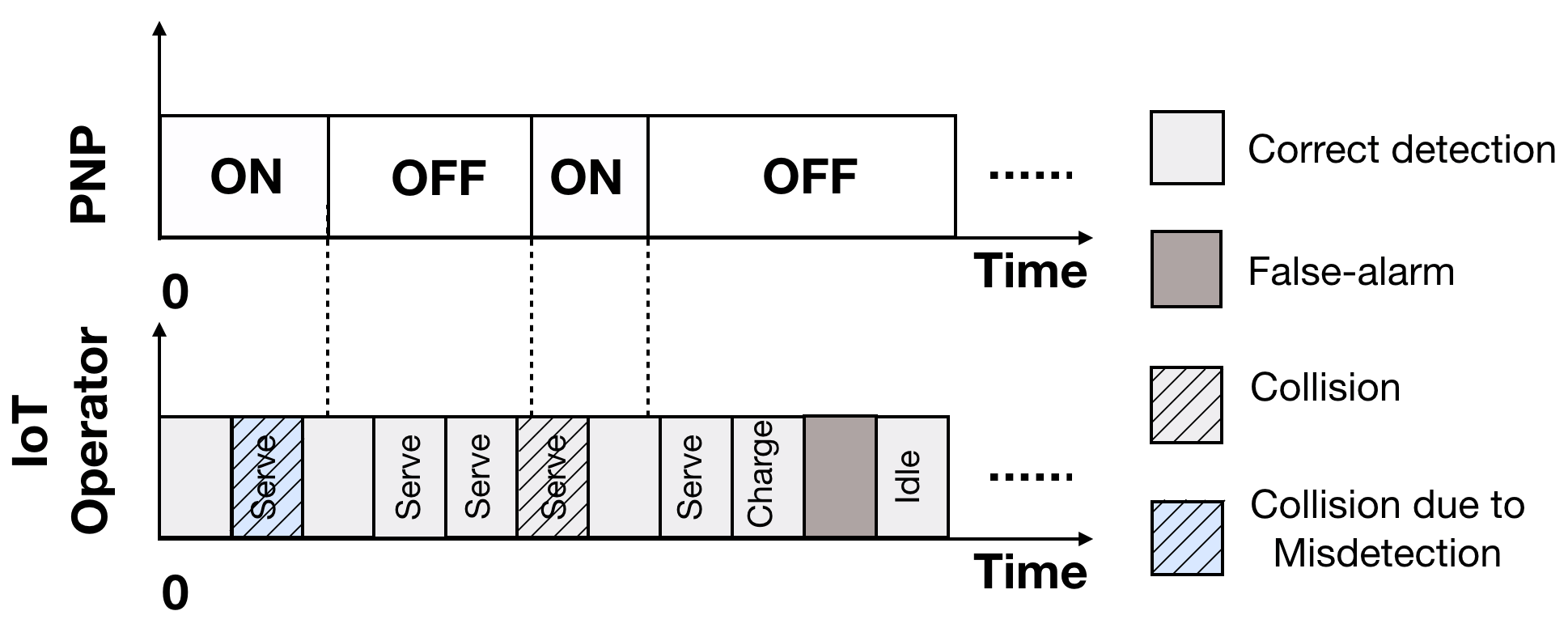}
	\caption{Activity of PNP and IoT operator}
	\label{Fig_Activities}
\end{figure}
In this paper, the PNP's activity is modeled as a sequence of ON and OFF periods which, similar to \cite{wang2010delay, wang2012characterization, liang2010delay}, are assumed to follow two distinct exponential distributions. Needless to say, the boundaries defining the activity of the PNP are completely arbitrary. In contrast, the IoT operator's activity must follow its spectrum sensing process which is scheduled to be performed by its sensors at regular intervals of times [Fig. \ref{Fig_Activities}]. Note that, the sensing process is prone to false-negative errors and can occasionally detect the PNP to be inactive even in the presence of its activity. Moreover, due to the arbitrariness of the boundaries of its ON periods, the PNP can resuscitate its activity in the middle of two consecutive sensing events, even if it were inactive at the first sensing epoch. Both of these scenarios have been shown by diagonally lined shaded regions in Fig. \ref{Fig_Activities}. Correct detection and false-alarm situations have been also depicted in the same figure. 

As a result of these sensing imperfections, the IoT operator might end up jeopardizing the PNP's intended transmissions by a significant amount if it decides to utilize each detected white-space to either transmit data or to harvest energy via MPT. In this paper, we circumvent this problem by forcing the IoT operator to remain idle for a certain fraction of white spaces even when it has sufficient capability to utilize it. The percentage of idle white-spaces works as a control parameter that can be used to tune the right balance between the QoS of the IoT network and the amount of interference inflicted on the PNP's users. We use the Discrete Time Markov Chain (DTMC) based framework to accurately analyze the above scenario.

\subsection{Our Contributions}
In this work, we consider a CR-IoT network, consisting of a Primary Network Provider (PNP) and an IoT operator, running in opportunistic mode. We assume that the IoT operator uses the MPT method for energy harvesting. For this purpose, the IoT operator deploys several Power Beacons (PBs) throughout the network. Our contributions can be summed up as follows:
 
 1. Due to the slotted sensing by the IoT operator, it is not possible to monitor PNP's activity if it changes in between two consecutive spectrum-sensing points. We take into account such possibilities and model the framework using DTMC based method.
 
 2. PNP imposes an interference limit to prevent its transmission from being jeopardized by the IoT operator's activity in case of misdetection by its CR sensors. We introduce the concept of idleness decision of the PBs that allows the IoT operator more flexibility to keep the interference within the constraint.
 
 3. Each PB of the IoT operator has three possible actions - remaining idle, charging IoT nodes, and collecting sensed data. Using our method, it is possible to obtain the optimal usage of PNP's inactivity period for performing these actions.
 
 4. Our model can determine the influences on the various QoS metrics of the CR-IoT network due to the energy harvesting process via the MPT method.
 
 5. When all the system parameters are known, we can determine the sustainability region (satisfying QoS requirements of the IoT operator and the interference limit demanded by PNP) of the CR-IoT network.

\subsection{Literature Survey}
In the existing literature, quite a few works on the analysis of CR-IoT systems are available. These works attempt to evaluate various QoS parameters (such as - age of information, delay, packet drop probability, jitter, etc.) of the CR-IoT systems. Some of these works also address the energy harvesting issue and utilize ambient sources (primarily PNP's transmission) for this purpose. MPT method has the advantage over energy harvesting via an ambient source in terms of energy efficiency and long-range transmission. Few works have been done analyzing MPT harvesting techniques in IoT systems. To the best of our knowledge, QoS analysis of a CR-IoT system that employs MPT has not been performed in the literature.

In \cite{liang2010delay}, the authors carried out the delay analysis of an opportunistic Cognitive Radio Sensor Network (CRSN) supporting real-time traffic. The authors of \cite{wang2010delay} and \cite{wang2012characterization} inspected the delay of a CR-IoT network and proposed an adaptive algorithm to control the packet generation rate of secondary users that minimizes queue length. In \cite{li2018network}, the authors obtained a unique set of optimal sensing parameters, that allows an opportunistic narrowband CR-IoT network to attain maximum throughput. In all of these works, the activities of the PNP and the IoT operator are assumed to be mutually exclusive. In a practical scenario, these activities are not synchronized and may get overlapped. In our work, we address this issue and develop a DTMC-based method that models the actual system quite accurately.

Due to the long-range and better energy efficiency of the MPT mechanism, it is the potential to become the predominant energy harvesting method in IoT systems. Several works have been done which tackle different problems related to energy harvesting via the MPT method. The authors of \cite{choi2018distributed} developed a distributed wireless power transfer system to improve the power transfer efficiency. In \cite{wang2019truthful}, the authors examined an IoT system comprised of several IoT operators and a single energy provider. They devised a mechanism to enforce truthfulness regarding the channel gains of the IoT nodes belonging to competing operators.
\section{System Model}
\begin{figure}
	\centering
	\includegraphics[width=80 mm]{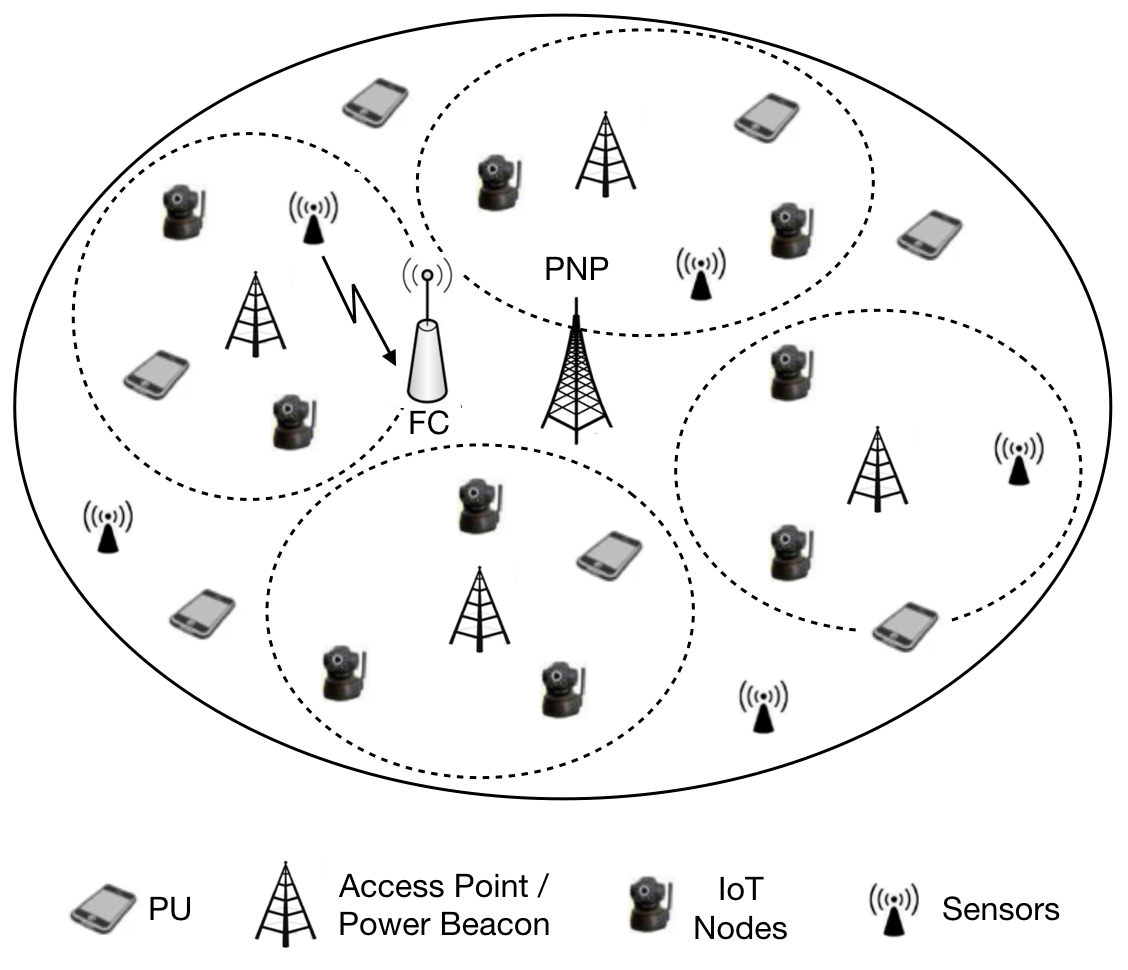}
	\caption{A schematic diagram of the system model}
	\label{Fig_System_Model}
\end{figure}
We analyze a wireless powered Cognitive Radio based Internet-of-Things (CR-IoT) network formed between a Primary Network Provider (PNP) and an IoT operator in an area $\mathcal{A}\subset\mathbb{R}^2$ [Fig. \ref{Fig_System_Model}]. PNP is assumed to have the ownership of the spectrum in that particular area. In our framework, we consider the CR-IoT network to operate in opportunistic mode, i.e. the IoT operator is allowed to access PNP's licensed spectrum in exchange for monetary payment, only during the inactivity period of PNP. To monitor PNP's activity, the IoT operator installs Cognitive Radio (CR) sensors throughout the area. It also deploys a Fusion Center (FC) for processing the sensing data collected by the CR sensors and takes the final decision regarding PNP's activity. The IoT operator deploys IoT nodes/sensors in the area to collect information for a certain application. To accumulate the data collected by the IoT nodes, several Access Points (APs) are put in the area by the IoT operator. IoT nodes have low battery life, so they must be recharged occasionally. Replenishing the batteries of the IoT nodes is the responsibility of the APs. So they also function as Power Beacons (PBs). In the rest of the paper, the terms Access Point (AP) and Power Beacon (PB) will be used interchangeably.

\subsection{System Description}
PNP's base-station is located at the center of the network, and its users are distributed in $\mathcal{A}$ following a Poisson Point Process. We consider only the downlink communication between PNP's base-station and its users. IoT nodes are allowed to transmit data to the APs or receive power from PB for wireless charging, only during the inactivity period of PNP's base-station. IoT nodes usually have low battery life. So their energy must be replenished from time to time. In our work, PBs are assumed to operate in Microwave Power Transfer (MPT) model for recharging the IoT nodes. Due to propagation loss, MPT cannot successfully recharge an IoT node beyond a certain distance. It should also be noted that a certain power threshold must be maintained at an IoT node to successfully recharge it. The propagation loss of the MPT method ensures negligible mutual interference on the transmission/charging process of IoT nodes belonging to different APs \cite{huang2014enabling}.

Each IoT node is associated with the nearest PB. Data packets generated at the IoT sensors are sent back to the associated AP, from which the IoT operator gathers information for application purposes. IoT nodes are generally small devices, which are not equipped to perform complex processing and do not have simultaneous wireless information-and-power transfer technology. As a result, the recharging process and the data transmission from IoT nodes to AP are mutually exclusive events.

IoT nodes have limited storage capability. So if its buffer becomes full, data packets may be dropped due to a lack of opportunity to transmit data to its AP. It is assumed that A cluster of IoT nodes belonging to an AP cannot interfere with the data transmission process of IoT nodes associated with other APs, due to the large distance between them and also due to low transmission power capability. For a given AP, a balance between charging and transmission must be maintained, while satisfying the QoS requirement of the IoT operator, and also the amount of interference incurred on PNP's transmission by the IoT operator's activity.

\subsection{Traffic Description}\label{section_iib}
Activity of PNP's base-station can be modelled as a sequence of ON and OFF intervals \cite{wang2010delay}. During the ON periods, PNP's base-station transmits to its users. Both the ON and OFF periods of PNP are assumed to be exponentially distributed \cite{wang2010delay}, and their distributions are denoted by $f_{\text{on}}(t)=\mu_{\text{on}} e^{-\mu_{\text{on}} t}$ and  $f_{\text{off}}(t)=\mu_{\text{off}} e^{-\mu_{\text{off}} t}$ respectively. Complementary CDF of ON and OFF intervals are represented by $F_{\text{on}}(t)$ and $F_{\text{off}}(t)$ respectively. The activity factor of the PNP is given by
\begin{align}\label{activity_factor}
\beta=\frac{\mathbb{E}\left[f_{\text{on}}(t)\right]}{\mathbb{E}\left[f_{\text{off}}(t)\right]+\mathbb{E}\left[f_{\text{on}}(t)\right]}=\frac{\mu_{\text{off}}}{\mu_{\text{on}}+\mu_{\text{off}}}
\end{align}

Consider a particular AP, which has $n$ IoT nodes inside its charging radius. The IoT operator performs sensing at a regular interval of $d$, which we define as a \textit{time-slot}. At the beginning of each slot, the IoT operator makes a decision regarding PNP's activity, using the sensing information collected by the CR sensors. If the IoT operator perceives PNP to be inactive, the AP can use that slot for clearing data from one of its nodes, or it can recharge the IoT nodes. Let $\xi \in (0,1)$ be the probability that a particular slot is used for charging. The minimum power level that must be maintained at an IoT node for successfully recharging it, is denoted by $P_{c}$. If PB decides to use a slot for data transmission, it selects an IoT node uniformly randomly and clears a single packet from it. It is assumed that a single slot duration is required to clear a packet.

Detection and false-alarm probabilities of the IoT operator are denoted by $P_D$ and $P_F$ respectively. In case of misdetection, PNP's transmission gets adversely affected by the IoT network's activity. We introduce the concept of idleness, which allows the IoT operator to reduce the probability of interfering with PNP's transmission. Even when the IoT network perceives PNP to be inactive, the AP may decide to neither charge nor clear the packet, while maintaining its QoS requirements. Idleness probability is denoted by $\theta \in [0,1]$. Our goal is to find the charging probability ($\xi$), idleness probability ($\theta$), and also the optimal power level ($P_{\text{th}}$) that should be maintained by the PB at each of its IoT nodes.

The data generation process at each IoT node is assumed to follow a Poisson distribution with mean $\lambda$. Each node has a finite buffer of length $K$ where it can store generated data. Service times of the $i^{th}$ IoT node follow the continuous distribution $b_i(t)$. Service time of an IoT node is the time required to clear one packet from the buffer. Service time distributions of the IoT nodes are assumed to be mutually independent. All IoT nodes are identical in terms of functionality. So, it is natural that each node gets an equal opportunity to transmit its data.

\subsection{Representative Sensor}
We consider a \textit{representative} (or \textit{equivalent}) IoT node $T$, which is equivalent to all the sensors (combined) under a particular AP. Aggregate data generation rate at the representative sensor follows Poisson distribution with mean $n\lambda$. Buffer-length of $T$ is $K$, same as the individual IoT node. Probability of $k$ arrivals to $T$ in a time interval $t$ is denoted as $A(t,k)=\frac{e^{-n\lambda t}(n\lambda t)^k}{k!}$, and the inter arrival time distribution is given by $a(t)=n\lambda e^{-n\lambda t}$. Clearing a packet from $T$ requires a single slot duration ($d$). We can model the representative sensor as a Discrete Time Markov Chain (DTMC).

\section{DTMC Analysis}
At the beginning of a slot (at any particular time $t$), the number of packets left in the system is denoted by $i$. PNP's activity at that instant of time is denoted by $\phi \in \{0,1\}$. $\phi = 0$ means PNP is inactive at that point of time and similarly, $\phi = 1$ denotes PNP's ongoing transmission. Depending upon the sensing information and idleness probability, AP can be in either of the three modes, namely Idle ($0$), service ($1$), and charging ($2$), in the slot starting at $t$. AP's decision regarding its activity in that slot is denoted by $\psi \in \{0,1,2\}$. After the completion of this slot (at time $t+d$), the number of packets remaining in the system is denoted by $j$. So, the state description at time $t$ is $(i,\phi,\psi)$. Before proceeding to calculate various transition probabilities among the states, we need to inspect the probability of transition of PNP's activity at the beginning and the end of a slot.

\subsection{Transition Probability Regarding PNP's Activity}
Given that the PNP is in the ON period at the beginning of a slot (at time $t$), the probability that PNP will be in ON or OFF period at the end of that slot (at time $t+d$), are denoted by $\alpha_{11}$ and $\alpha_{10}$ respectively. So, clearly $\alpha_{10}=1-\alpha_{11}$. Similarly, the probability of PNP being in the OFF of ON period at the end of the slot, given that, it was in the OFF period at the beginning of the slot, are given by $\alpha_{00}$ and $\alpha_{01}=1-\alpha_{00}$. We shall describe the method to calculate $\alpha_{11}$, when ON and OFF period distributions are known.
\begin{figure}
	\centering
	\includegraphics[width=87 mm]{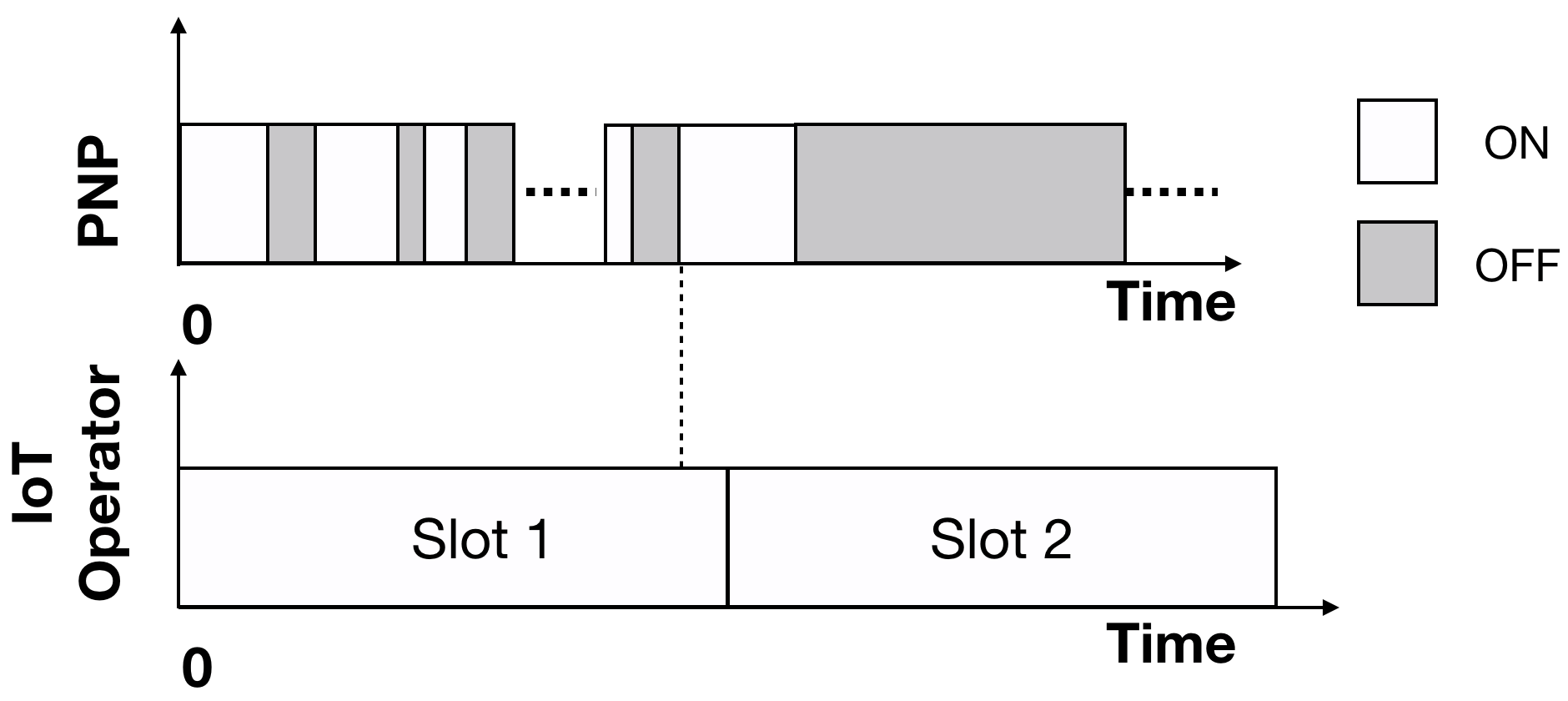}
	\caption{PNP is active at the beginning and end of a slot}
	\label{Fig_Alpha_Calc}
\end{figure}
If PNP remains in ON period at both the beginning and end of a slot, that means either a single ON period encompasses the whole slot, or a sequence of ON-OFF intervals must occur with a final ON period that includes the endpoint of the slot [Fig. \ref{Fig_Alpha_Calc}]. As both the ON and OFF distributions are exponential in nature, they have memoryless property. Probability that a single ON period of PNP encompasses a slot is given by $F_{\text{on}}(d)=e^{-\mu_{\text{on}}d}$.

Now we shall consider the other case (sequence of ON-OFF periods with final ON interval that includes the slot-end). Let $U_k$ be the random variable that describes the duration of $k$ consecutive pairs of ON-OFF intervals. Laplace transform of ON and OFF interval distributions are denoted by $F_{\text{on}}(s)$ and $F_{\text{off}}(s)$ respectively. Laplace transform of the distribution followed by $U_k$ can be written as $H_k(s)=[F_{\text{on}}(s)F_{\text{off}}(s)]^k$. These $k$ pairs of ON-OFF intervals must start at the beginning of a slot, and finish before the completion of that slot. Finally, another ON period occurs encompassing the end of the time-slot. The probability of such an event (denoted by $R_k$) can be written as
\begin{align}
R_k=\int\limits_{t=0}^{d}\mathscr{L}^{-1}\left\{H_k(s)\right\}(t) ~F_{\text{on}}(d-t)~\mathrm{d}t
\end{align}
The total probability of PNP being in ON mode at the end of the slot, given that it was in ON mode at the beginning of the slot, can be written as
\begin{align}
\begin{split}
\alpha_{11}&=e^{-\mu_{\text{on}}d}+\sum_{k=1}^{\infty}R_k\\
&=e^{-\mu_{\text{on}}d}+\int\limits_{t=0}^{d}\mathscr{L}^{-1}\left\{\sum_{k=1}^{\infty}H_k(s)\right\}e^{-\mu_{\text{on}}(d-t)}\\
&=e^{-\mu_{\text{on}}d}+\int\limits_{t=0}^{d}\mathscr{L}^{-1}\left\{\frac{F_{\text{on}}(s)F_{\text{off}}(s)}{1-F_{\text{on}}(s)F_{\text{off}}(s)}\right\}e^{-\mu_{\text{on}}(d-t)}\mathrm{d}t\\
&=e^{-\mu_{\text{on}}d}+\int\limits_{t=0}^{d}\frac{\mu_{\text{on}}\mu_{\text{off}}\left(1-e^{-\left(\mu_{\text{on}}+\mu_{\text{off}}\right)t}\right)}{\mu_{\text{on}}+\mu_{\text{off}}} e^{-\mu_{\text{on}}(d-t)}\mathrm{d}t\\
&=\frac{\mu_{\text{off}}+\mu_{\text{on}}e^{-\left(\mu_{\text{on}}+\mu_{\text{off}}\right)d}}{\mu_{\text{on}}+\mu_{\text{off}}}
\end{split}
\end{align}
We can also calculate $\alpha_{00}$ in a similar fashion. We can write down the rest of the probabilities as
\begin{align}
\begin{split}
\alpha_{00}&=\frac{\mu_{\text{on}}+\mu_{\text{off}}e^{-\left(\mu_{\text{on}}+\mu_{\text{off}}\right)d}}{\mu_{\text{on}}+\mu_{\text{off}}}\\
\alpha_{10}&=\frac{\mu_{\text{on}}-\mu_{\text{on}}e^{-\left(\mu_{\text{on}}+\mu_{\text{off}}\right)d}}{\mu_{\text{on}}+\mu_{\text{off}}}\\
\alpha_{01}&=\frac{\mu_{\text{off}}-\mu_{\text{off}}e^{-\left(\mu_{\text{on}}+\mu_{\text{off}}\right)d}}{\mu_{\text{on}}+\mu_{\text{off}}}
\end{split}
\end{align}

\subsection{State Transition Probabilities}
In this subsection, we derive various state transition probabilities. Set of all possible transitions can be divided into few classes, which are described in the subsequent sub-subsections. Few transitions in each sub-subsection may require special attention. Those are described in details, and the rest can be derived using similar idea. Explicit expressions of all the transition probabilities can be found in the appendix. Source and destination states are represented by $(i,\phi,\psi)$ and $(j,\phi',\psi')$ respectively. Probability of going to state $(j,\phi',\psi')$ from state $(i,\phi,\psi)$ is denoted by $\pi_{(j,\phi',\psi')}^{(i,\phi,\psi)}$. Probability of generation of $k$ new packets during a slot duration is given by $A(d,k)=\frac{e^{-n\lambda d}(n\lambda d)^k}{k!}$.

\subsubsection{Transitions from $(0,\phi,\psi)$ to $(0,\phi',\psi')$}\label{invalid_states}
If the system is devoid of packets at the beginning of a slot, under no circumstances AP will go into service mode. As a result, $(0,0,1)$ and $(0,1,1)$ are invalid states. Now, let us consider the transition from $(0,\phi,0/2)$ to $(0,0,0)$. By $(0,\phi,0/2)$, we mean that the state can be either $(0,\phi,0)$ or $(0,\phi,2)$. In the destination state, PNP is in OFF mode. So AP's idleness may occur due to three cases - due to false alarm [probability $P_F$], idleness decision [probability $(1-P_F)\theta$] or service decision without any remaining packet [probability $(1-P_F)(1-\theta)(1-\xi)$]. We also take into account that PNP moved from state $\phi$ to $0$ and no packet is generated during the slot duration. So, the transition probability can be written as
\begin{align}
\begin{split}
\pi_{(0,0,0)}^{(0,\phi,0/2)}&=\alpha_{\phi 0}\left[P_F+(1-P_F)\theta+ \right. \\
&\left. (1-P_F)(1-\theta)(1-\xi)\right]A(d,0)
\end{split}
\end{align}

\subsubsection{Transitions from $(0,\phi,\psi)$ to $(j,\phi',\psi')~\forall j \in [1,K]$}
Similar to the previous sub-subsection, source state cannot be either $(0,0,1)$ or $(0,1,1)$. Consider the transition from $(0,\phi,0/2)$ to $(j,0,1)$. As the destination state has $j(>0)$ number of packets left, there is a chance of AP going into the service mode, unlike the previous sub-subsection. A service decision by AP implies no false-alarm by the CR sensors (correctly identifying PNP's inactivity), as well as no idleness and charging decision. If the buffer is not full ($j<K$) in the destination state $j$ arrivals happen in the slot duration. $K$ or more arrivals will make the buffer full at the destination state. As PNP goes from mode $\phi$ to OFF interval, we can write the transition probability as
\begin{align}
\begin{split}
\pi_{(j,0,1)}^{(0,\phi,0/2)}&=\alpha_{\phi 0}(1-P_F)(1-\xi)(1-\theta)A(d,j)~\forall j\in [1,K)\\
\pi_{(K,0,1)}^{(0,\phi,0/2)}&=\alpha_{\phi 0}(1-P_F)(1-\xi)(1-\theta)\sum_{k=K}^{\infty}A(d,k)
\end{split}
\end{align}

\subsubsection{Transitions from $(i,\phi,\psi)$ to $(j,\phi',\psi')$, $i>0$,~$\max(i-1,1)\leq j < K-1$}\label{subsubsection_3}
A single slot duration can clear at most one packet from the system. So the number of packets in the destination cannot go down by more than one compared to the source state. $(1,\phi,\psi)$ to $(0,\phi',\psi')$ transition requires special attention, which is done in the next sub-subsection.

Now, let us inspect the transition from $(i,0,1)$ to $(j,0,0)$, where $i>0$ and $j\geq \max(i-1,1)$. It is given that, at the source state AP decided to go into service mode. That service is successfully executed if and only if the OFF interval that started at the source state continues for the whole slot duration, probability of which is given by $F_{\text{off}}(d)$. In this case, $j-i+1$ new packets are required to have $j$ packets in the destination state. The service is unsuccessful if at least one ON interval occurs within the slot duration, which has the probability $\alpha_{00}-F_{\text{off}}(d)$, and only $j-i$ new packets are necessary to have $j$ packets at the destination. At the destination state, AP decides to remain idle, which may happen due to two scenarios - either it has a false-alarm case [probability $P_F$] or it decides to go into idle period despite correct sensing [probability $(1-P_F)\theta$]. Finally, we can write down the transition probability as
\begin{align}
\begin{split}
\pi_{(j,0,0)}^{(i,0,1)}&=\left[P_F+(1-P_F)\theta\right]\left\{F_{\text{off}}(d)A(d,j-i+1) \right.\\
& \left. +(\alpha_{00}-F_{\text{off}}(d))A(d,j-i)\right\}
\end{split}
\end{align}

\subsubsection{Transitions from $(1,\phi,\psi)$ to $(0,\phi',\psi')$}
As we have already mentioned earlier, destination cannot be one of the two invalid states [namely $(0,0,1)$ and $(0,1,1)$]. In addition to that, if PNP is in ON mode at the source state, a packet cannot be cleared successfully from the system. If PNP is in OFF mode at the source and ON at the destination state, successful service is again impossible. So, we can only have the following transitions.
\begin{align}
\begin{split}
\pi_{(0,0,0)}^{(1,0,1)}&=\left[P_F+(1-P_F)\theta+(1-P_F)(1-\theta)(1-\xi)\right]\\
&\left\{F_{\text{off}}(d)A(d,0)\right\}\\
\pi_{(0,0,2)}^{(1,0,1)}&=\left[(1-P_F)(1-\theta)\xi\right]\left\{F_{\text{off}}(d)A(d,0)\right\}
\end{split}
\end{align}

\subsubsection{Transitions from $(i,\phi,\psi)$ to $(K-1,\phi',\psi')$, $i>0$}
Consider the system when it comes from a source state, where service decision was taken. If that service is successful, no more than $K-1$ packets may remain in the system. For example, let us analyse the transition from $(i,0,1)$ to $(K-1,0,1)$, where $i>0$. Similar to \ref{subsubsection_3}, successful service happens with probability $F_{\text{off}}(d)$ and service is interrupted with probability $\alpha_{00}-F_{\text{off}}(d)$. In case of completion of service, $K-i$ or more new packets will take the system-buffer to $K-1$ packets, whereas only $K-i-1$ packets are required in case of unsuccessful service. As PNP is off at the destination state, service decision is taken with probability $(1-P_F)(1-\theta)(1-\xi)$. The transition probability can now be written as
\begin{align}
\begin{split}
\pi_{(K-1,0,1)}^{(i,0,1)}&=(1-P_F)(1-\xi)(1-\theta)\left\{F_{\text{off}}(d)\sum_{k=K-i}^{\infty}A(d,j) \right. \\
& \left. +(\alpha_{00}-F_{\text{off}}(d))A(d,K-i-1)\right\}
\end{split}
\end{align}

\subsubsection{Transitions from $(i,\phi,\psi)$ to $(K,\phi',\psi')$, $i>0$}
Now, we are left with only one special case. To reach a full buffer, no successful service can happen in the time-slot, and $K-i$ or more packets must be generated. For example, consider the case of transition from $(i,0,1)$ to $(K,0,1)$. Probability of unsuccessful service is $\alpha_{00}-F_{\text{off}}(d)$, as shown in \ref{subsubsection_3}. Choosing service mode in destination state happens with probability $(1-P_F)(1-\theta)(1-\xi)$. So, overall transition probability can be written as
\begin{align}
\pi_{(K,0,1)}^{(i,0,1)}\hspace*{-0.3em}=(1-P_F)(1-\xi)(1-\theta)\left(\alpha_{00}-F_{\text{off}}(d)\right)\hspace*{-0.5em}\sum_{k=K-i}^{\infty}\hspace*{-0.4em}a(d,k)
\end{align}

\subsection{Steady State Probability}
We have described all possible state transitions of our system in the previous subsection. Now, we can move on to the steady-state probabilities of the representative sensor. In section \ref{invalid_states}, we have mentioned that a state $(i,\phi,\psi)$ can not be of form $(0,0,1)$ or $(0,1,1)$, as the system never resumes service when there is no packet is remaining in the queue. When there are a non-zero number of packets left in the queue, $\phi$ and $\psi$ can take all possible values. The maximum number of packets that can be accommodated in the queue is $K$. So, the total number of valid states of our system is $6(K+1)-2=6K+4$. Let $\boldsymbol{\mu}$ be the vector representing the steady-state probabilities of being in these states. The transition matrix, denoted by \textbf{P}, can be populated using the method described in the previous subsection. The stationary distribution of the system can be found by solving $\boldsymbol{\mu} \textbf{P} = \boldsymbol{\mu}$.

\section{Performance Analysis of the System}
In the preceding section, we have demonstrated a method to analyze the functionality of an AP along with the associated IoT nodes. In this section, we try to quantify the performance of this system via various QoS metrics, such as - packet drop probability and the average waiting time of the generated data at IoT nodes, and also the probability of interference on PNP's transmission, due to misdetection by the IoT operator. Probability of the system being in the state $(i,\phi,\psi)$ is denoted by $P_{(i,\phi,\psi)}$.

\subsection{Packet Drop Probability}
The packet drop probability (denoted by $P_B$) of a queuing system is intimately related to the carried load of the system. Carried load (denoted by $\rho_c$) of a queuing system is defined as the probability that the server is busy at an arbitrary point in time. For our system, the whole timeline is divided into slots of duration $d$. So, the carried load is equivalent to the probability that a particular slot is successfully used for clearing a packet. For a successful service, two conditions must be satisfied, PNP should be inactive at the beginning of a slot, and PNP also does not resume transmission during that slot. So, carried load can be formally written as $\rho_c = e^{-\mu_{\text{off}}d}\sum_{i=1}^{K}P_{(i,0,1)}$. Utilization factor of the queue is given by $\rho = n\lambda/d$. So, the packet drop probability of the system can be written as
\begin{align}
P_B = 1 - \frac{\rho_c}{\rho} = 1 - \frac{de^{-\mu_{\text{off}}d}}{n\lambda}\sum_{i=1}^{K}P_{(i,0,1)}
\end{align}

\begin{figure*}[h]
\begin{subfigure}{0.33\textwidth}
\includegraphics[width=1\linewidth]{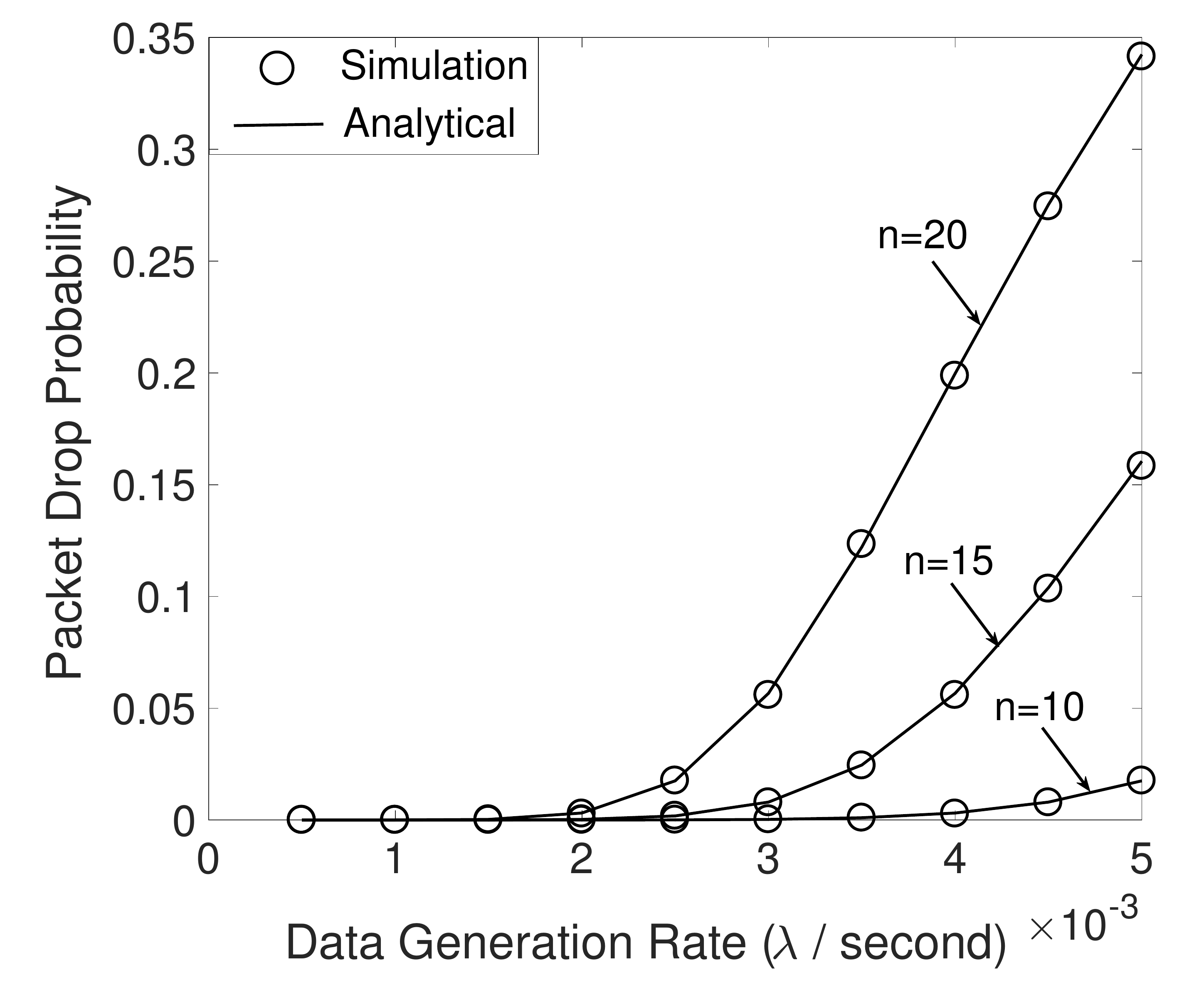}
\caption{Data Rate vs Packet Drop Probability}
\label{subfig_datarate_vs_PB}
\end{subfigure}
\begin{subfigure}{0.33\textwidth}
\includegraphics[width=1\linewidth]{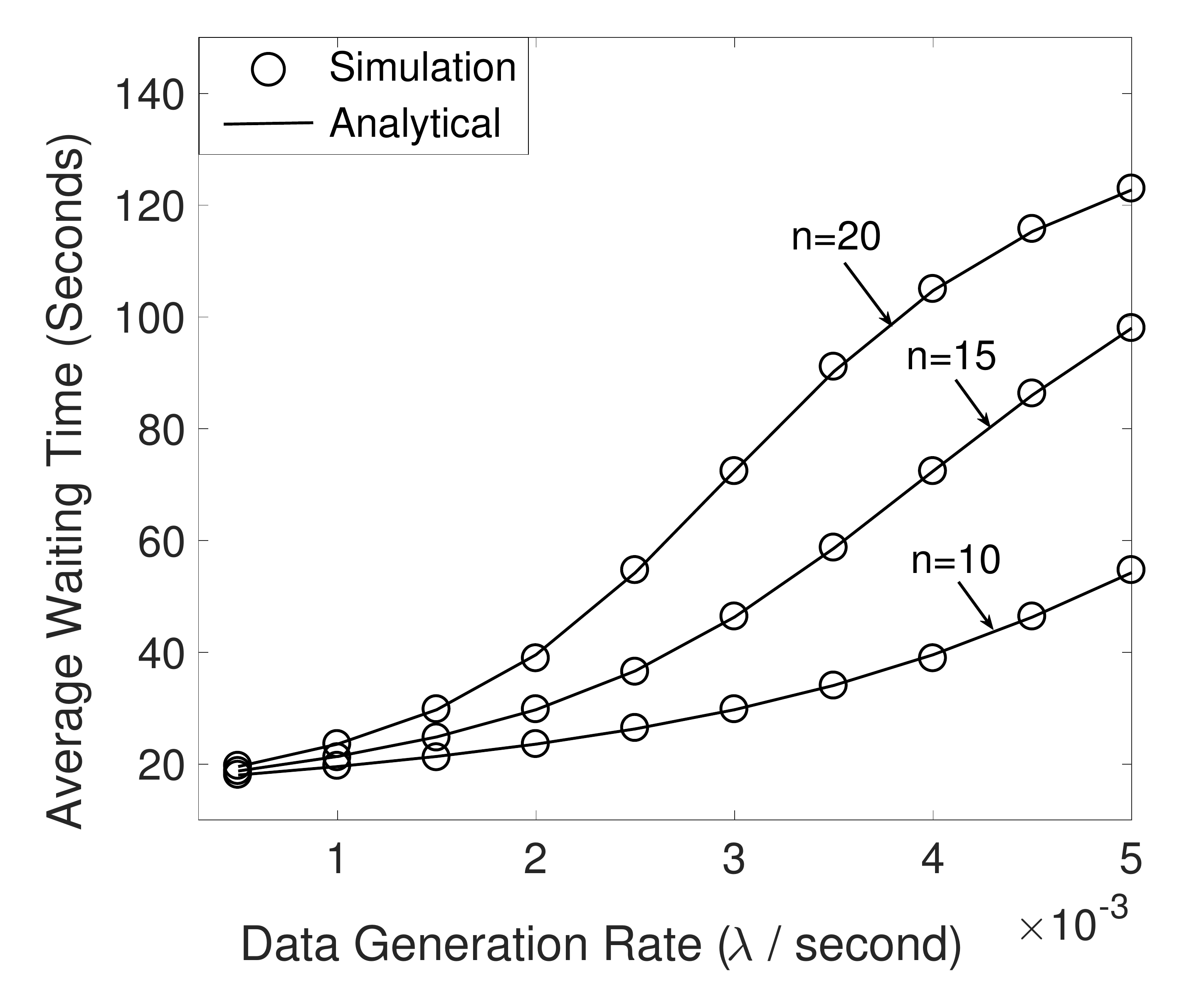}
\caption{Data Rate vs Average waiting time}
\label{subfig_datarate_vs_delay}
\end{subfigure}
\begin{subfigure}{0.33\textwidth}
\includegraphics[width=1\linewidth]{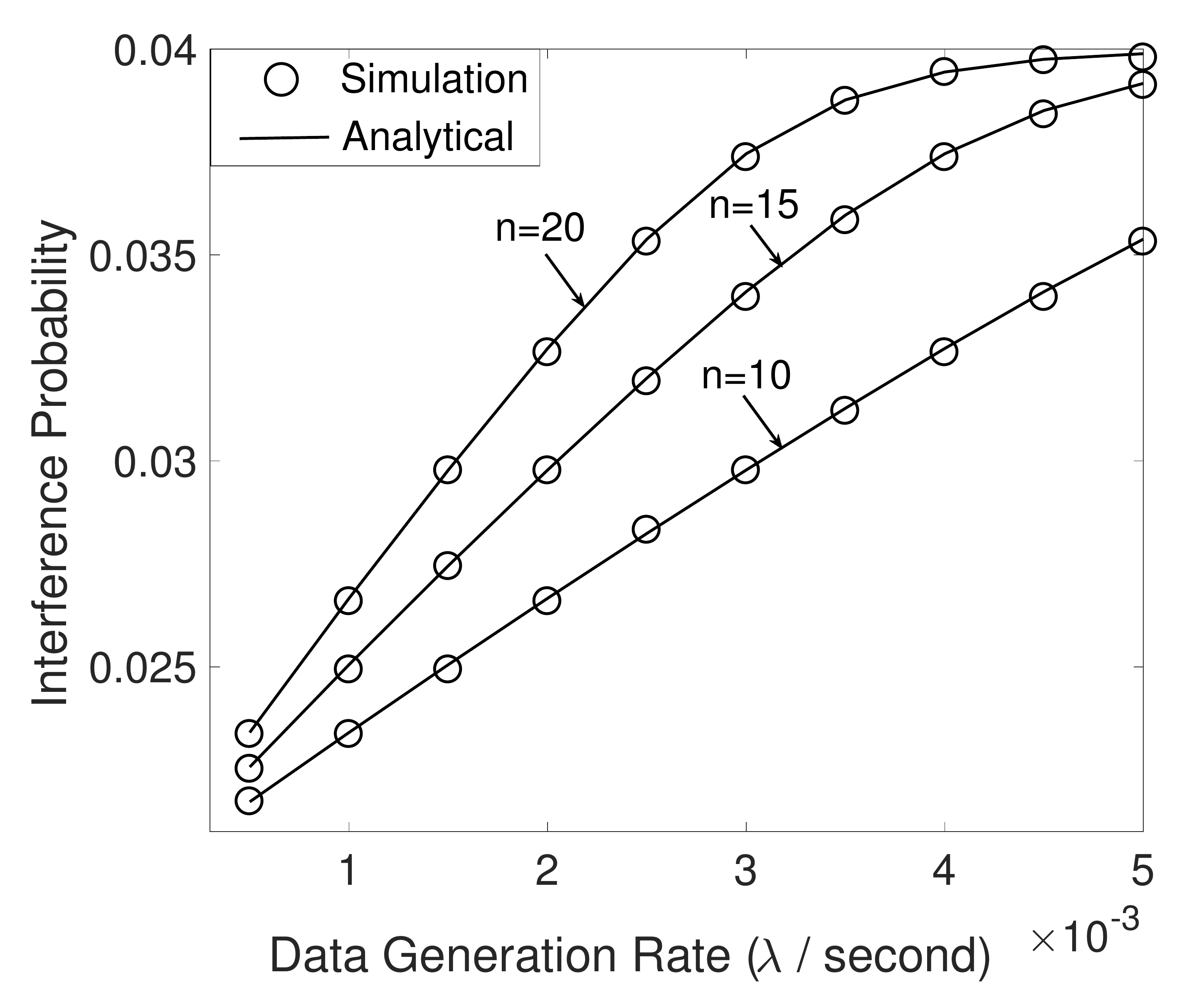}
\caption{Data Rate vs Interference Probability}
\label{subfig_datarate_vs_IP}
\end{subfigure}
\caption{QoS Metrics of the CR-IoT Network w.r.t. Data Generation Rate at the IoT Nodes}	
\end{figure*}

\subsection{Average Waiting Time}
Using Little's law, the average waiting time (denoted by $W$) of a queuing system can be calculated as $W=L/\lambda_{\text{eff}}$. $L$ is the average number of packets in the system. Due to the finiteness of the queue, some packets get dropped when the queue is full. So, the effective packet generation rate at the representative sensor becomes $\lambda_{\text{eff}}=n\lambda (1-P_B)$. Now, we aim to derive the average number of packets ($L$) in the system. Let  $N(t)$ be the number of packets in the queue at time $t$.

When a newly generated packet is successfully included in the queue, the probability that the system already had $i \in [0,K-1]$ packets, is denoted by $\gamma_i$. After a new packet arrival, the probability of finding the system with $i \in [0,K]$ packets, irrespective of whether the new packet can or can not get in, is given by $\epsilon_i$. After successful packet transmission, the probability that the system has $i$ packets left is represented by $\delta_i$. From the well-known PASTA property (Poisson Arrivals See Time Averages), for any queuing process with Poisson arrival, the probability of finding the system having $i$ packets at any arbitrary instant is exactly equal to $\epsilon_i$ \cite{baccelli2006role}. Number of packets in system changes due to either arrival or departure of a packet \cite{kleinrock1975queueing}, so clearly $\gamma_i=\delta_i~\forall ~ i \in [0,K-1]$.
\begin{lemma}\label{thinning_lemma}
$\forall ~ i \in [0,K-1]$, $\epsilon_i=(1-P_B)\gamma_i$ and $\epsilon_K=P_B$.
\end{lemma}
\begin{proof}
$\forall ~ i \in [0,K-1]$,
\begin{align}
\begin{split}
\gamma_i&=\lim_{t \to \infty}\lim_{\Delta t\to 0}P[N(t)=k|N(t)\hspace*{-0.2em}<\hspace*{-0.2em}K,\text{Arrival in }(t,t+\Delta t)]\\
&=\lim_{t \to \infty}\lim_{\Delta t\to 0}\frac{P[N(t)=k,N(t)\hspace*{-0.2em}<\hspace*{-0.2em}K,\text{Arrival in }(t,t+\Delta t)]}{P[N(t)<K,\text{Arrival in }(t,t+\Delta t)]}\\
&=\lim_{t \to \infty}\lim_{\Delta t\to 0}\frac{P[N(t)=k|\text{Arrival in }(t,t+\Delta t)]}{P[N(t)<K|\text{Arrival in }(t,t+\Delta t)]}\\
&=\frac{\epsilon_i}{1-P_B}
\end{split}
\end{align}
If a new packet is generated while the buffer of the queue is already full, it needs to be dropped. This is exactly the definition of the packet drop probability, so $\epsilon_K=P_B$. This concludes the proof.
\end{proof}
Now the final step for evaluating average waiting time is to calculate $\delta_i \forall ~ i \in [0,K-1]$, which in turn will help to find out $\gamma_i$s and $\epsilon_i$s. In order to have $i$ packets after a successful service, the system must have at least $1$ packet (otherwise service decision has no meaning) and at most $i+1$ packets (only one packet is cleared during a slot-duration) at the beginning of the slot. Successful service also requires PNP to be inactive at the beginning of the slot and also during the slot-duration. Given all these conditions are satisfied, probability of finding the system with $i$ packets after a successful transmission is denoted by $\kappa_i$. We need to consider the states where PNP is inactive ($\phi=0$), and AP decides to clear packet ($\psi=1$). So, $\forall ~ i \in [0,K-1]$, $\kappa_i= e^{-\mu_{\text{off}}d}\sum_{j=1}^{i+1}P_{(j,0,1)}$. We can normalize these probabilities to find out $\delta_i$s. So, $\delta_i=\kappa_i/\sum_{i=0}^{K-1}\kappa_i ~\forall ~ i \in [0,K-1]$. Final expression of the average waiting time is
\begin{align}
W=\frac{\sum_{i=0}^{K}\epsilon_i}{n\lambda (1-P_B)}=\frac{P_B}{n\lambda (1-P_B)}+\frac{1}{n\lambda}\sum_{i=0}^{K-1}\delta_i
\end{align}

\subsection{Interference Probability}
Due to misdetection by the IoT operator, APs will interfere with PNP's transmission, either by charging IoT nodes ($\psi=2$) or by receiving sensed data from them ($\psi=1$). The probability of such interference (denoted by $P_I$) can be evaluated as
\begin{align}
P_I = P_{(0,1,2)}+\sum_{i=1}^{K}\left(P_{(i,1,1)}+P_{(i,1,2)}\right)
\end{align}

\subsection{Optimal Power Level for Microwave Power Transfer}
In this subsection, we present a method to find out the total transmission power of the APs. Consider one of the APs under the IoT operator. The number of IoT sensors associated with it is denoted by $n$. The $i^{th}$ sensor is at a distance $r_i$ from the its PB. The minimum power level for successfully charging an IoT node is $P_c$. The maximum transmission power that can be provided by the AP is $P_m$. Now, examine a time duration $\tau$ after the CR-IoT network formation. The total number of packets generated at an IoT node during this time is given by $\lambda (1-P_B)\tau$. Time available to AP for charging its sensor (when neither transferring data nor being idle) is $(1-\beta)(1-\theta)\xi\tau$.  Clearing a single packet on average requires $m$ unit of energy. To balance demand of energy, power level required at the $i^{th}$ sensor can be calculated as $P_{\text{th}}^i=\max\left(P_c,\frac{m\lambda (1-P_B)}{(1-\beta)(1-\theta)\xi}\right)$. So, the total power that must be provided by the PB is given by $P_{\text{th}}=\min \left(P_m,\sum_{i \in \text{PB}} P_{\text{th}}^i r_i^{\alpha}\right)$, where $\alpha$ is the pathloss coefficient. The network will be unstable if the required power exceeds the maximum power level ($P_m$) that can be provided by an AP.

\begin{figure*}[ht]
	\centering
	\begin{subfigure}{0.24\textwidth}
		\includegraphics[width=\linewidth]{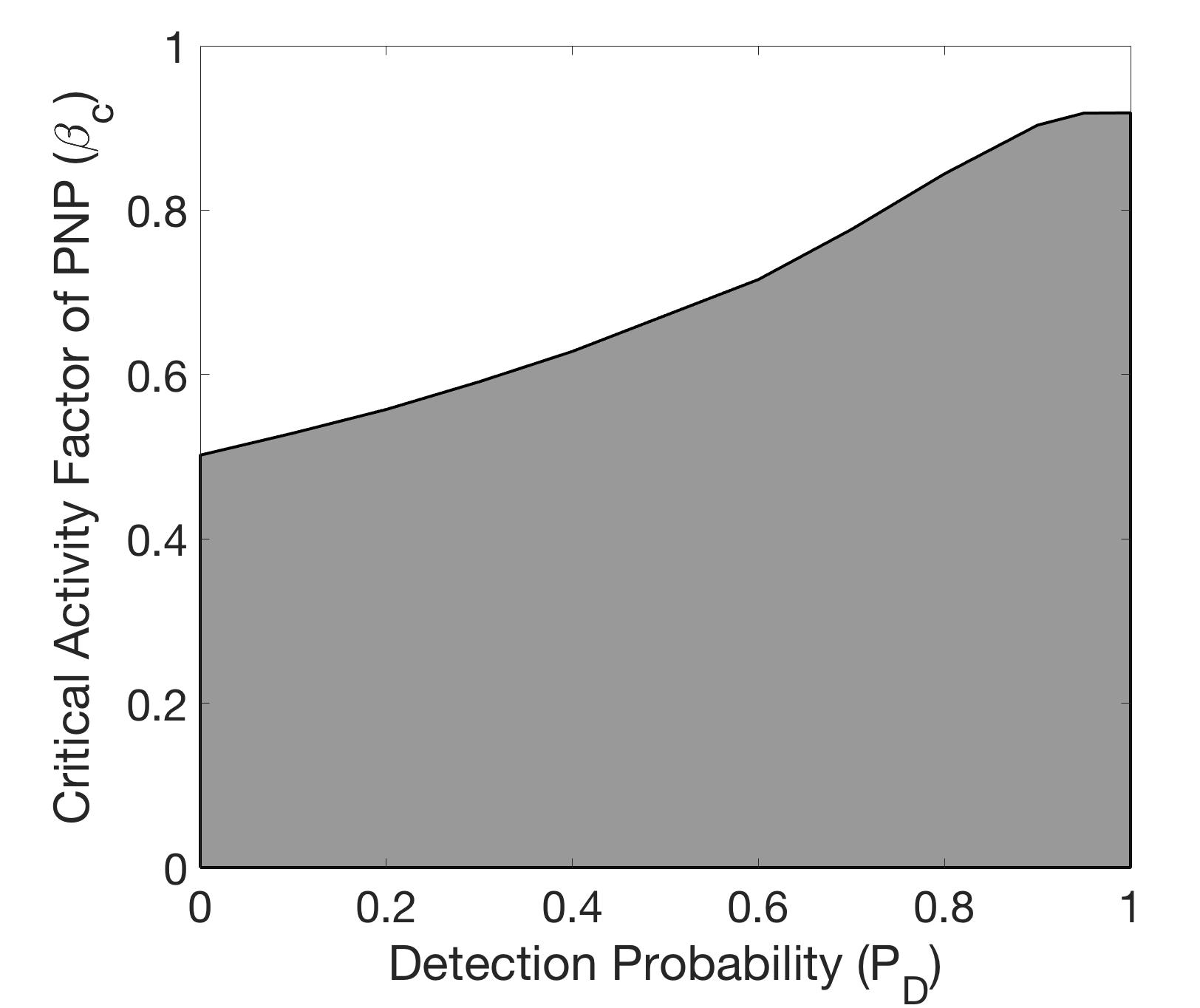}
		\caption{$P_D$ vs $\beta_c$}
		\label{subfig_PD_vs_beta}
	\end{subfigure}
	\begin{subfigure}{0.24\textwidth}
		\includegraphics[width=\linewidth]{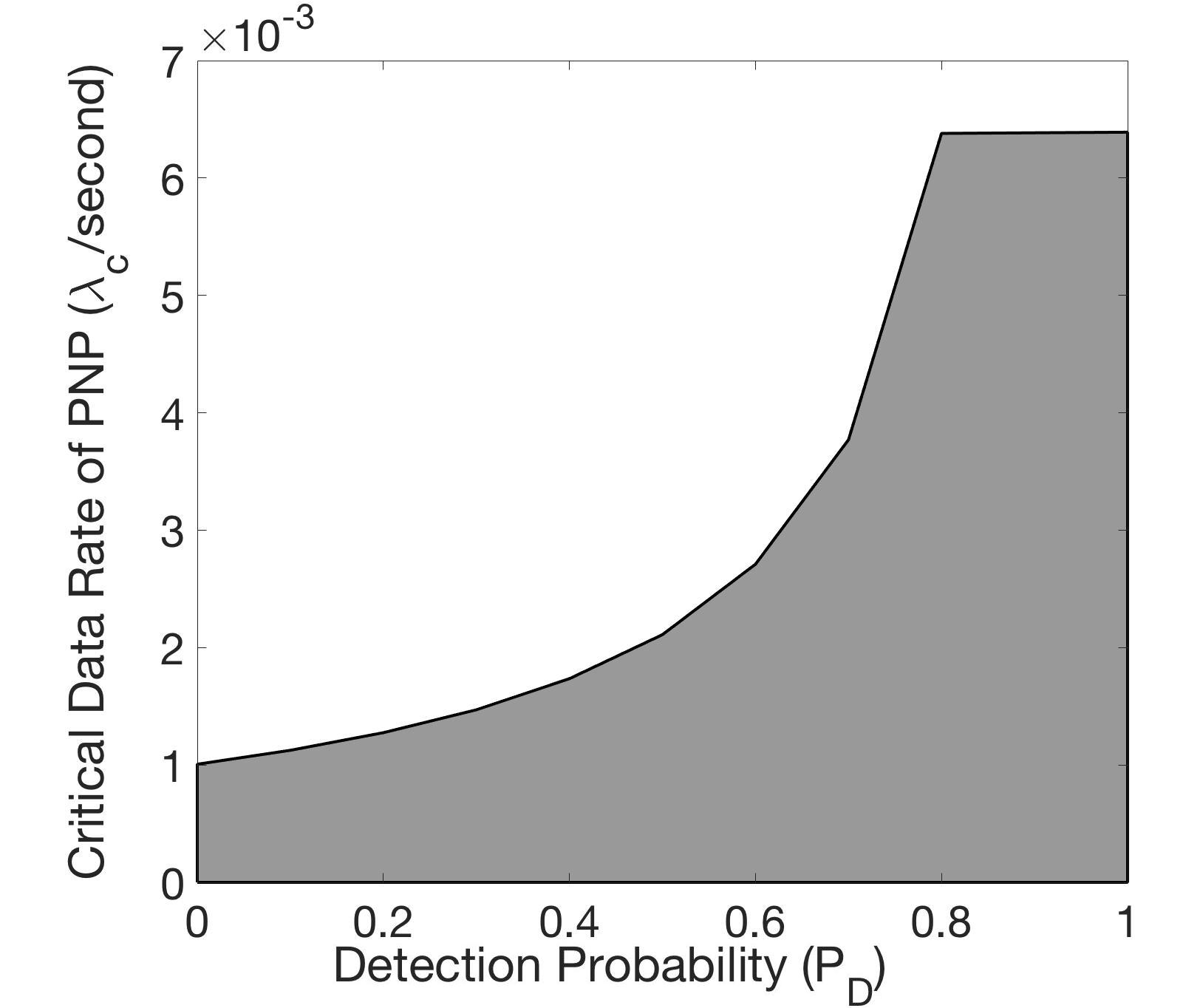}
		\caption{$P_D$ vs $\lambda_c$}
		\label{subfig_PD_vs_data}
	\end{subfigure}
	\centering
	\begin{subfigure}{0.24\textwidth}
		\includegraphics[width=\linewidth]{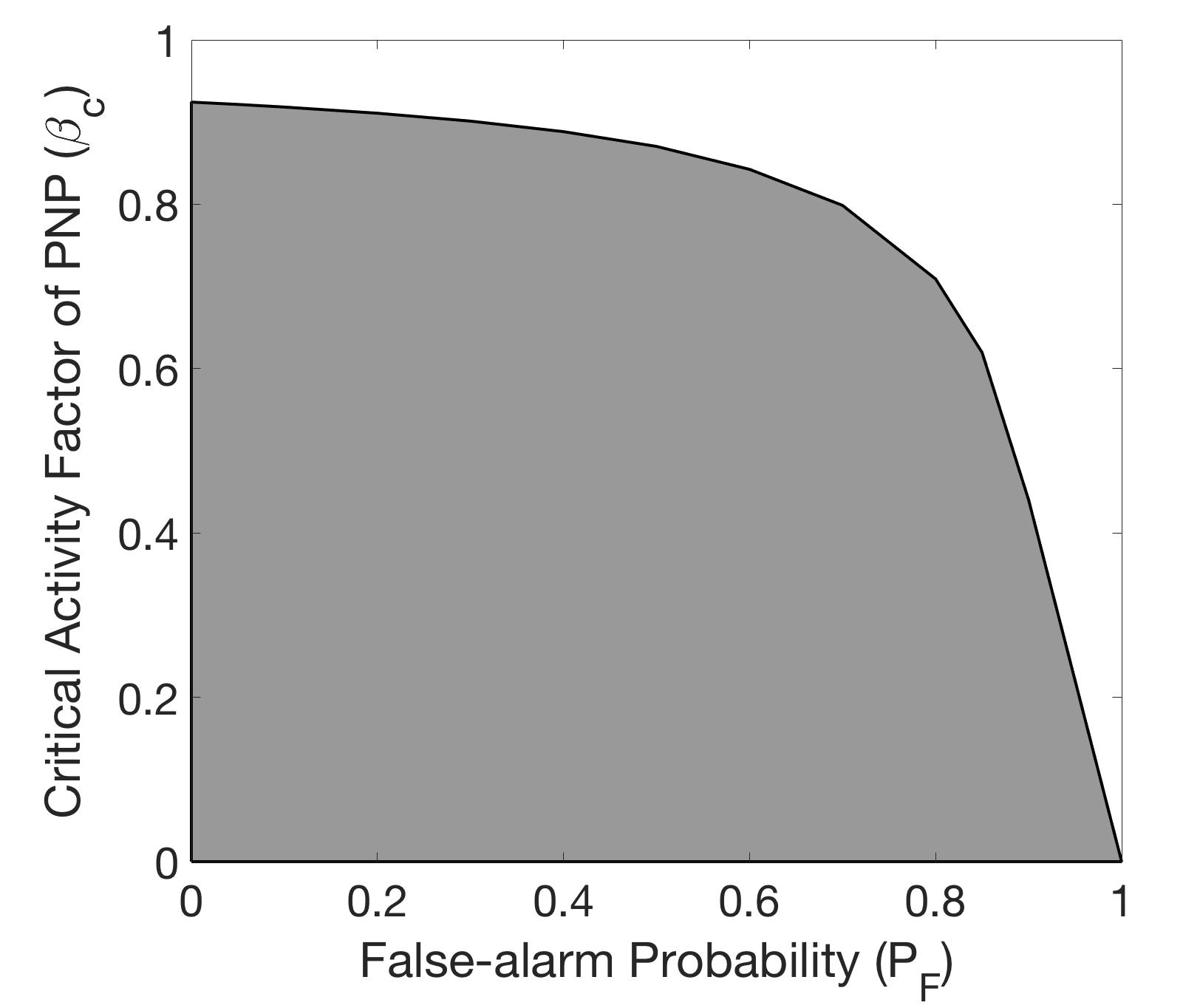}
		\caption{$P_F$ vs $\beta_c$}
		\label{subfig_PF_vs_beta}
	\end{subfigure}
	\begin{subfigure}{0.24\textwidth}
		\includegraphics[width=\linewidth]{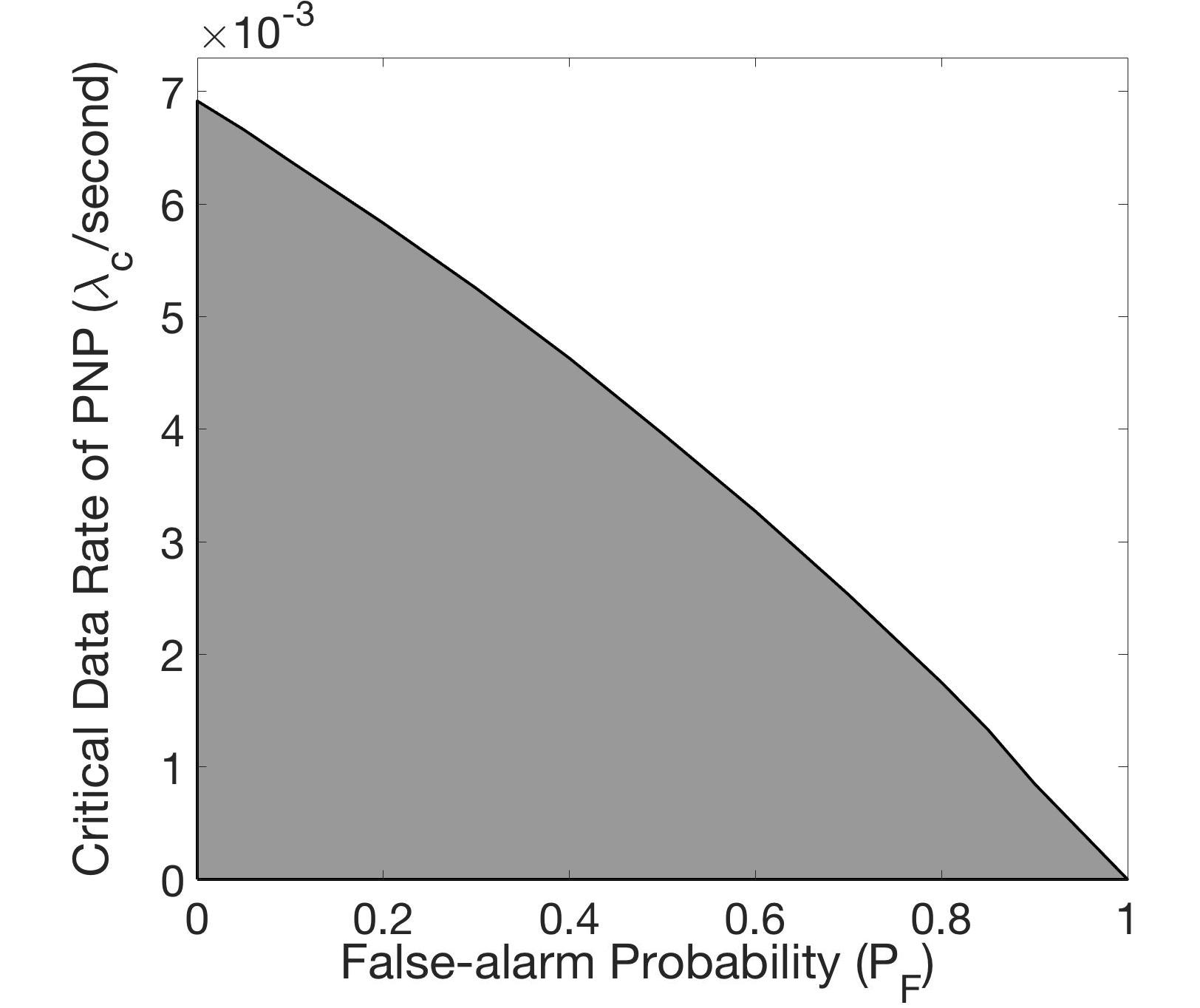}
		\caption{$P_F$ vs $\lambda_c$}
		\label{subfig_PF_vs_data}
	\end{subfigure}
	\caption{Sustainability region of the CR-IoT}
\end{figure*}

\section{Simulation results}\label{section_v}
In this section, we investigate how different system parameters affect the QoS of the CR-IoT operator, and also the chance of interfering with PNP's transmission, due to misdetection by the IoT operator. We also find out the sustainability region of this CR-IoT system. Before moving on to the simulation results, we shall briefly describe the simulation setup. The region of operation of the CR-IoT system is a $10$ km $\times$ $10$ km square area. PNP's base-station is placed at its center and its users are uniformly distributed over the area. There are $10$ APs under the IoT operator, each with a charging radius of $1$ km. The maximum transmission power ($P_m$) of each AP is $10$ W, and the minimum threshold power that must be retained at each IoT node for successful charging is $50$ $\mu$W \cite{janhunen2019wireless}. For a particular AP, the number of associated IoT nodes follow Poisson Distribution and they are uniformly distributed within the AP's charging boundary. The buffer size ($K$) of each IoT node is $10$. Packet transmission time distribution (service time) of each IoT sensor follows an exponential distribution with parameter $1$ packet/second. The average energy required for transmitting one packet is $400$ $\mu$Joule \cite{janhunen2019wireless}.

As in \cite{wang2010delay}, PNP's ON and OFF intervals follow exponential distributions, namely $f_1(t)=\mu_{\text{on}} e^{-\mu_{\text{on}} t}$ and $f_0(t)=\mu_{\text{off}} e^{-\mu_{\text{off}} t}$. Duration of each time slot, at the beginning of which the IoT operator makes the decision about PNP's activity, is taken as $1$ second. The simulation is performed using OMNeT++ software, and the analytical comparison is done in Python. In the following subsections, we show that results obtained from our theoretical framework match very closely with the real setup.

\subsection{Variation of QoS with data generation rate at the IoT nodes}
In this subsection, we inspect few key performance metrics of the CR-IoT network, namely the packet drop probability ($P_B$), the average waiting time ($W$) of the data generated at the IoT nodes, and the interference probability ($P_I$) on PNP's transmission, with respect to the data generation rate. For this simulation, charging probability ($\xi$) and idleness probability ($\theta$) are taken as $0.5$ and $0.2$ respectively. Also, the packet drop and the false-alarm probability of the IoT operator are assumed to be $0.9$ and $0.1$ respectively. PNP's activity factor ($\beta$), which has been defined in section \ref{section_iib}, is set to be $0.5$. With an increase in packet generation rate, the buffer at the IoT nodes is likely to be filled up quickly. As a result, a new data packet will more frequently find the buffer to be full, thus increasing the overall packet drop probability of the network [Fig. \ref{subfig_datarate_vs_PB}]. As the idleness and charging probability remain fixed, data transmission opportunities remain the same. So, increased data generation rate will culminate in a longer waiting time for a packet [Fig. \ref{subfig_datarate_vs_delay}]. With an increase in data generation rate, the IoT operator has to increase both data transfer and charging frequency, which in turn enhances the chance of interfering with PNP's activity [Fig. \ref{subfig_datarate_vs_IP}].

\subsection{Sustainability region of the CR-IoT network with variation of detection probability ($P_D$)}
In this subsection, we show the sustainability region of the CR-IoT system with respect to the detection probability ($P_D$) of the IoT operator. Proper functionality of all the APs are required for the successful operation of the IoT operator. We consider a particular AP under the IoT operator. The number of sensors associated with this AP is $20$. To sustain the CR-IoT system, the IoT operator must maintain its QoS requirements and also the interference limit imposed by the license owner (PNP). For the simulation setup, both the packet drop probability ($P_B$) and interference probability ($P_I$) requirements are set at $0.1$. The false-alarm probability of the IoT operator is kept at $0.1$. 

In Fig. \ref{subfig_PD_vs_beta}, we show the effect of PNP's activity ($\beta$) on the viability of our system. For that purpose, the data generation rate at each IoT node is assumed to be $0.001$ packet/second. It is clear that, with the improvement in detection performance by the IoT operator, it can withstand higher PNP activity without violating the QoS and interference constraints. For a given detection probability (even the perfect detection), if PNP's activity exceeds a certain critical level ($\beta_c$), it is impossible to sustain the network due to a lack of enough correctly detected white-space. An increase in PNP's activity means less opportunity for IoT operator's activity. So, a higher level of transmission power is needed to quickly recharge the IoT nodes. Beyond the $\beta_c$ value, it is not possible to sustain the network without surpassing the maximum transmission power level of the AP.

Next, we keep the activity factor of PNP fixed at $0.5$, and obtain the maximum possible data generation rate of the IoT nodes that is possible under the given constraints. As in the earlier case, improved detection allows increased data generation rate of the IoT nodes [Fig. \ref{subfig_PD_vs_data}]. For a particular detection probability, the system is infeasible beyond a critical data rate ($\lambda_c$ packet/second) due to the physical limitation of the maximum transmission power of the AP.

\subsection{Sustainability region of the CR-IoT network with variation of false-alarm probability ($P_F$)}
In this subsection, false-alarm probability ($P_F$) of the IoT operator is varied while its detection probability ($P_D$) is kept at $0.9$. Same as in the previous subsection, both $P_B$ and $P_I$ requirement are $0.1$.

Firstly, we keep the data generation rate of each IoT node at $0.001$ packet/second. As the false-alarm probability increases, chances of leaving out white-spaces increase, so PNP's activity needs to go down for maintaining the CR-IoT network [Fig. \ref{subfig_PF_vs_beta}]. For a particular value of $P_F$, PNP's activity must not exceed a critical level ($\beta_c$), otherwise, the IoT operator cannot remain operational without increasing its APs' maximum transmission power capability. If $P_F$ of the IoT operator becomes $1$, no amount of transmission power of the AP can keep the system afloat even at zero activity of the PNP.

Finally, we inspect the effect of the data generation rate of the IoT node on our system's functionalities. PNP's activity is assumed to be $0.5$. An increase in the $P_F$ value will allow a lesser data rate of the IoT nodes, as the IoT operator misses out on white-space because of its degrading false-alarm probability [Fig. \ref{subfig_PF_vs_data}]. If the sensing performance of the IoT operator and other system parameters are known, we can calculate the maximum data rate of the IoT nodes ($\lambda_c$ packet/second) that can sustain the CR-IoT network.

\subsection{Effect of synchronicity assumption of PNP and IoT operator's activities}
\begin{figure}[h]
\centering
\includegraphics[width=0.8\linewidth]{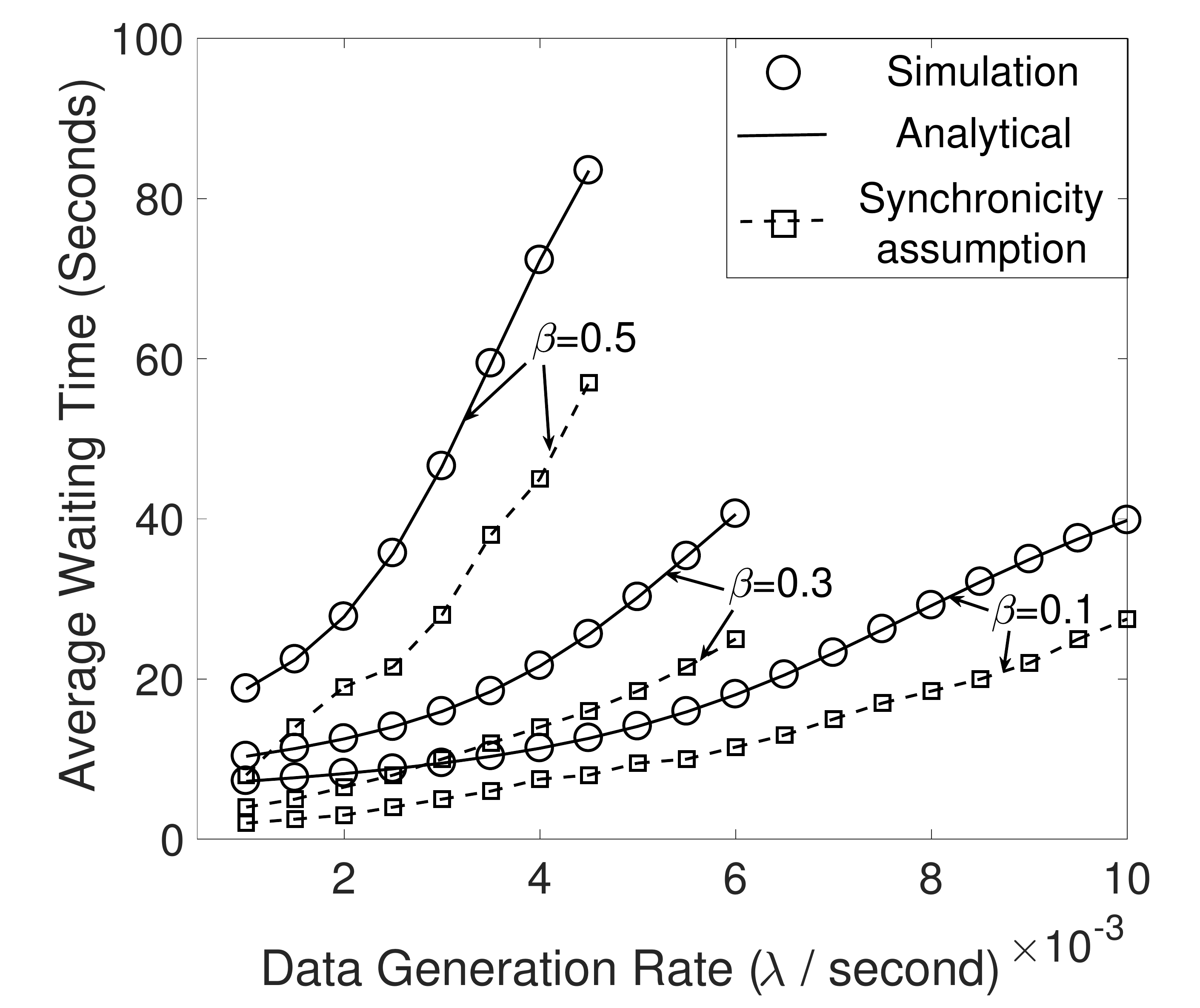}
\caption{Effect of synchronicity assumption}
\label{fig_comparison}
\end{figure}
In \cite{wang2012characterization}, the authors assumed that the PNP's ON and OFF durations are much larger than the duration of a single time-slot of the IoT operator. So they neglected the possibility of collisions between PNP and IoT operator's activity. In other words, the IoT operator is expected to know PNP's activity perfectly, making their activities synchronized. In a practical scenario, occasional overlaps between PNP and IoT operator's activities are bound to happen. In our work, we remove the synchronicity assumption. For the simulation purpose, the AP is assumed to be connected to $20$ IoT sensors. For avoiding collision due to misdetection, the IoT operator's detection performance has to be completely flawless. So, the value of $P_D$ is set to be $1$. To obtain best QoS performance, the false-alarm probability ($P_F$) is set at $0$. In Fig. \ref{fig_comparison}, we plot the average waiting time of the data generated at the IoT nodes with respect to the data generation rate for three cases - the OMNeT++ simulation, our analysis and the analysis done at \cite{wang2012characterization}. The results obtained by the authors of \cite{wang2012characterization} are over-optimistic and deviate significantly from the simulation results of a CR-IoT system [Fig. \ref{fig_comparison}]. Results obtained from our model agree accurately with the OMNeT++ simulation of a real CR-IoT system [Fig. \ref{fig_comparison}].

\section{Conclusions and Future Works}
In this work, we analyze a CR-IoT network operating in opportunistic mode, that enables various smart city facilities. Due to the energy efficiency of the MPT, the IoT operator employs this method for energy harvesting. We propose an analytical method that allows the IoT operator to maintain the balance between different actions (data transfer from IoT sensors, recharging, and staying idle) while satisfying various QoS requirements and complying with the interference limit enforced by the license owner (PNP). As shown in section \ref{section_v}, our theoretical results match very well with the actual simulation of the CR-IoT system. Our model can find the sustainability region of such a CR-IoT system if all the system parameters are known. Availability of multiple licensed band will introduce an additional action of the IoT operator namely choosing suitable spectrum for its own activities. This can be a possible direction for future work.

\appendix

\section*{Transitions from $(0,\phi,\psi)$ to $(0,\phi',\psi')$}
\begin{align}
\begin{split}
\pi_{(0,0,0)}^{(0,\phi,0/2)}&=\alpha_{\phi 0}\left[P_F+(1-P_F)\theta+ \right. \\
&\left. (1-P_F)(1-\theta)(1-\xi)\right]A(d,0)\\
\pi_{(0,0,2)}^{(0,\phi,0/2)}&=\alpha_{\phi 0}(1-P_F)\xi (1-\theta)A(d,0)\\
\pi_{(0,1,0)}^{(0,\phi,0/2)}&=\alpha_{\phi 1}\left[P_D+(1-P_D)\theta+ \right. \\
&\left. (1-P_D)(1-\theta)(1-\xi)\right]A(d,0)\\
\pi_{(0,1,2)}^{(0,\phi,0/2)}&=\alpha_{\phi 1}(1-P_D)\xi(1-\theta)A(d,0)
\end{split}
\end{align}

\section*{Transitions from $(0,\phi,\psi)$ to $(j,\phi',\psi')~\forall j \in [1,K-1]$}
\begin{align}
\begin{split}
\pi_{(j,0,0)}^{(0,\phi,0/2)}&=\alpha_{\phi 0}\left[P_F+(1-P_F)\theta\right]A(d,j)\\
\pi_{(j,0,1)}^{(0,\phi,0/2)}&=\alpha_{\phi 0}(1-P_F)(1-\xi)(1-\theta)A(d,j)\\
\pi_{(j,0,2)}^{(0,\phi,0/2)}&=\alpha_{\phi 0}(1-P_F)\xi (1-\theta)A(d,j)\\
\pi_{(j,1,0)}^{(0,\phi,0/2)}&=\alpha_{\phi 1}\left[P_D+(1-P_D)\theta\right]A(d,j)\\
\pi_{(j,1,1)}^{(0,\phi,0/2)}&=\alpha_{\phi 1}(1-P_D)(1-\xi)(1-\theta)A(d,j)\\
\pi_{(j,1,2)}^{(0,\phi,0/2)}&=\alpha_{\phi 1}(1-P_D)\xi(1-\theta)A(d,j)
\end{split}
\end{align}

\section*{Transitions from $(0,\phi,\psi)$ to $(K,\phi',\psi')$}
\begin{align}
\begin{split}
\pi_{(K,0,0)}^{(0,\phi,0/2)}&=\alpha_{\phi 0}\left[P_F+(1-P_F)\theta\right] \sum_{k=K}^{\infty}A(d,k)\\
\pi_{(K,0,1)}^{(0,\phi,0/2)}&=\alpha_{\phi 0}(1-P_F)(1-\xi)(1-\theta)\sum_{k=K}^{\infty}A(d,k)\\
\pi_{(K,0,2)}^{(0,\phi,0/2)}&=\alpha_{\phi 0}(1-P_F)\xi (1-\theta)\sum_{k=K}^{\infty}A(d,k)\\
\pi_{(K,1,0)}^{(0,\phi,0/2)}&=\alpha_{\phi 1}\left[P_D+(1-P_D)\theta\right]\sum_{k=K}^{\infty}A(d,k)\\
\pi_{(K,1,1)}^{(0,\phi,0/2)}&=\alpha_{\phi 1}(1-P_D)(1-\xi)(1-\theta)\sum_{k=K}^{\infty}A(d,k)\\
\pi_{(K,1,2)}^{(0,\phi,0/2)}&=\alpha_{\phi 1}(1-P_D)\xi(1-\theta)\sum_{k=K}^{\infty}A(d,k)
\end{split}
\end{align}

\section*{Transitions from $(i,\phi,\psi)$ to $(j,\phi',\psi')$, $i>0$,~$\max(i-1,1)\leq j < K-1$}
\begin{align}
\begin{split}
\pi_{(j,0,0)}^{(i,\phi,0/2)}&=\alpha_{\phi 0}\left[P_F+(1-P_F)\theta\right] A(d,j-i)\\
\pi_{(j,0,1)}^{(i,\phi,0/2)}&=\alpha_{\phi 0}(1-P_F)(1-\xi)(1-\theta)A(d,j-i)\\
\pi_{(j,0,2)}^{(i,\phi,0/2)}&=\alpha_{\phi 0}(1-P_F)\xi (1-\theta)A(d,j-i)\\
\pi_{(j,1,0)}^{(i,\phi,0/2)}&=\alpha_{\phi 1}\left[P_D+(1-P_D)\theta\right]A(d,j-i)\\
\pi_{(j,1,1)}^{(i,\phi,0/2)}&=\alpha_{\phi 1}(1-P_D)(1-\xi)(1-\theta)A(d,j-i)\\
\pi_{(j,1,2)}^{(i,\phi,0/2)}&=\alpha_{\phi 1}(1-P_D)\xi(1-\theta)A(d,j-i)
\end{split}
\end{align}

\begin{align}
\begin{split}
\pi_{(j,1,0)}^{(i,\phi,1)}&=\alpha_{\phi 1}\left[P_D+(1-P_D)\theta\right]A(d,j-i)\\
\pi_{(j,1,1)}^{(i,\phi,1)}&=\alpha_{\phi 1}(1-P_D)(1-\xi)(1-\theta)A(d,j-i)\\
\pi_{(j,1,2)}^{(i,\phi,1)}&=\alpha_{\phi 1}(1-P_D)\xi(1-\theta)A(d,j-i)
\end{split}
\end{align}

\begin{align}
\begin{split}
\pi_{(j,0,0)}^{(i,0,1)}&=\left[P_F+(1-P_F)\theta\right]\left\{F_{\text{off}}(d)A(d,j-i+1) \right.\\
& \left. +(\alpha_{00}-F_{\text{off}}(d))A(d,j-i)\right\}\\
\pi_{(j,0,1)}^{(i,0,1)}&=(1-P_F)(1-\xi)(1-\theta)\left\{F_{\text{off}}(d)A(d,j-i+1) \right.\\
& \left. +(\alpha_{00}-F_{\text{off}}(d))A(d,j-i)\right\}\\
\pi_{(j,0,2)}^{(i,0,1)}&=(1-P_F)\xi (1-\theta)\left\{F_{\text{off}}(d)A(d,j-i+1) \right.\\
& \left. +(\alpha_{00}-F_{\text{off}}(d))A(d,j-i)\right\}
\end{split}
\end{align}

\begin{align}
\begin{split}
\pi_{(j,0,0)}^{(i,1,1)}&=\alpha_{10}\left[P_F+(1-P_F)\theta\right]A(d,j-i)\\
\pi_{(j,0,1)}^{(i,1,1)}&=\alpha_{10}(1-P_F)(1-\xi)(1-\theta)A(d,j-i)\\
\pi_{(j,0,2)}^{(i,1,1)}&=\alpha_{10}(1-P_F)\xi (1-\theta)A(d,j-i)
\end{split}
\end{align}

\section*{Transitions from $(1,\phi,\psi)$ to $(0,\phi',\psi')$}
\begin{align}
\begin{split}
\pi_{(0,0,0)}^{(1,0,1)}&=\left[P_F+(1-P_F)\theta+(1-P_F)(1-\theta)(1-\xi)\right]\\
&\left\{F_{\text{off}}(d)A(d,0)\right\}\\
\pi_{(0,0,2)}^{(1,0,1)}&=\left[(1-P_F)(1-\theta)\xi\right]\left\{F_{\text{off}}(d)A(d,0)\right\}
\end{split}
\end{align}

\section*{Transitions from $(i,\phi,\psi)$ to $(K-1,\phi',\psi')$, $i>0$}
\begin{align}
\begin{split}
\pi_{(K-1,0,0)}^{(i,\phi,0/2)}&=\alpha_{\phi 0}\left[P_F+(1-P_F)\theta\right] A(d,K-i-1)\\
\pi_{(K-1,0,1)}^{(i,\phi,0/2)}&=\alpha_{\phi 0}(1-P_F)(1-\xi)(1-\theta)A(d,K-i-1)\\
\pi_{(K-1,0,2)}^{(i,\phi,0/2)}&=\alpha_{\phi 0}(1-P_F)\xi (1-\theta)A(d,K-i-1)\\
\pi_{(K-1,1,0)}^{(i,\phi,0/2)}&=\alpha_{\phi 1}\left[P_D+(1-P_D)\theta\right]A(d,K-i-1)\\
\pi_{(K-1,1,1)}^{(i,\phi,0/2)}&=\alpha_{\phi 1}(1-P_D)(1-\xi)(1-\theta)A(d,K-i-1)\\
\pi_{(K-1,1,2)}^{(i,\phi,0/2)}&=\alpha_{\phi 1}(1-P_D)\xi(1-\theta)A(d,K-i-1)
\end{split}
\end{align}

\begin{align}
\begin{split}
\pi_{(K-1,0,0)}^{(i,0,1)}&=\left[P_F+(1-P_F)\theta\right]\left\{F_{\text{off}}(d)\sum_{k=K-i}^{\infty}A(d,k) \right. \\
& \left. +(\alpha_{00}-F_{\text{off}}(d))A(d,K-i-1)\right\}\\
\pi_{(K-1,0,1)}^{(i,0,1)}&=(1-P_F)(1-\xi)(1-\theta)\left\{F_{\text{off}}(d)\sum_{k=K-i}^{\infty}A(d,j) \right. \\
& \left. +(\alpha_{00}-F_{\text{off}}(d))A(d,K-i-1)\right\}\\
\pi_{(K-1,0,2)}^{(i,0,1)}&=(1-P_F)\xi (1-\theta)\left\{F_{\text{off}}(d)\sum_{k=K-i}^{\infty}A(d,j)+ \right. \\
& \left. (\alpha_{00}-F_{\text{off}}(d))A(d,j-i)\right\}
\end{split}
\end{align}

\begin{align}
\begin{split}
\pi_{(K-1,1,0)}^{(i,\phi,1)}&=\alpha_{\phi 1}\left[P_D+(1-P_D)\theta\right]A(d,K-i-1)\\
\pi_{(K-1,1,1)}^{(i,\phi,1)}&=\alpha_{\phi 1}(1-P_D)(1-\xi)(1-\theta)A(d,K-i-1)\\
\pi_{(K-1,1,2)}^{(i,\phi,1)}&=\alpha_{\phi 1}(1-P_D)\xi(1-\theta)A(d,K-i-1)
\end{split}
\end{align}

\begin{align}
\begin{split}
\pi_{(K-1,0,0)}^{(i,1,1)}&=\alpha_{10}\left[P_F+(1-P_F)\theta\right]A(d,K-i-1)\\
\pi_{(K-1,0,1)}^{(i,1,1)}&=\alpha_{10}(1-P_F)(1-\xi)(1-\theta)A(d,K-i-1)\\
\pi_{(K-1,0,2)}^{(i,1,1)}&=\alpha_{10}(1-P_F)\xi (1-\theta)A(d,K-i-1)
\end{split}
\end{align}

\section*{Transitions from $(i,\phi,\psi)$ to $(K,\phi',\psi')$, $i>0$}
\begin{align}
\begin{split}
\pi_{(K,0,0)}^{(i,\phi,0/2)}&=\alpha_{\phi 0}\left[P_F+(1-P_F)\theta\right]\sum_{k=K-i}^{\infty}A(d,k)\\
\pi_{(K,0,1)}^{(i,\phi,0/2)}&=\alpha_{\phi 0}(1-P_F)(1-\xi)(1-\theta)\sum_{k=K-i}^{\infty}A(d,k)\\
\pi_{(K,0,2)}^{(i,\phi,0/2)}&=\alpha_{\phi 0}(1-P_F)\xi (1-\theta)\sum_{k=K-i}^{\infty}A(d,k)\\
\pi_{(K,1,0)}^{(i,\phi,0/2)}&=\alpha_{\phi 1}\left[P_D+(1-P_D)\theta\right]\sum_{k=K-i}^{\infty}A(d,k)\\
\pi_{(K,1,1)}^{(i,\phi,0/2)}&=\alpha_{\phi 1}(1-P_D)(1-\xi)(1-\theta)\sum_{k=K-i}^{\infty}A(d,k)\\
\pi_{(K,1,2)}^{(i,\phi,0/2)}&=\alpha_{\phi 1}(1-P_D)\xi(1-\theta)\sum_{k=K-i}^{\infty}A(d,k)
\end{split}
\end{align}

\begin{align}
\begin{split}
\pi_{(K,0,0)}^{(i,0,1)}\hspace*{-0.3em}&=\left[P_F+(1-P_F)\theta\right]\left(\alpha_{00}-F_{\text{off}}(d)\right)\hspace*{-0.5em}\sum_{k=K-i}^{\infty}A(d,k)\\
\pi_{(K,0,1)}^{(i,0,1)}\hspace*{-0.3em}&=(1-P_F)(1-\xi)(1-\theta)\left(\alpha_{00}-F_{\text{off}}(d)\right)\hspace*{-0.5em}\sum_{k=K-i}^{\infty}\hspace*{-0.4em}a(d,k)\\
\pi_{(K,0,2)}^{(i,0,1)}\hspace*{-0.3em}&=(1-P_F)\xi (1-\theta)\left(\alpha_{00}-F_{\text{off}}(d)\right)\hspace*{-0.5em}\sum_{k=K-i}^{\infty}a(d,k)
\end{split}
\end{align}

\begin{align}
\begin{split}
\pi_{(K,1,0)}^{(i,\phi,1)}&=\alpha_{\phi 1}\left[P_D+(1-P_D)\theta\right]\sum_{k=K-i}^{\infty}A(d,k)\\
\pi_{(K,1,1)}^{(i,\phi,1)}&=\alpha_{\phi 1}(1-P_D)(1-\xi)(1-\theta)\sum_{k=K-i}^{\infty}A(d,k)\\
\pi_{(K,1,2)}^{(i,\phi,1)}&=\alpha_{\phi 1}(1-P_D)\xi(1-\theta)\sum_{k=K-i}^{\infty}A(d,k)
\end{split}
\end{align}

\begin{align}
\begin{split}
\pi_{(K,0,0)}^{(i,1,1)}&=\alpha_{10}\left[P_F+(1-P_F)\theta\right]\sum_{k=K-i}^{\infty}A(d,k)\\
\pi_{(K,0,1)}^{(i,1,1)}&=\alpha_{10}(1-P_F)(1-\xi)(1-\theta)\sum_{k=K-i}^{\infty}A(d,k)\\
\pi_{(K,0,2)}^{(i,1,1)}&=\alpha_{10}(1-P_F)\xi (1-\theta)\sum_{k=K-i}^{\infty}A(d,k)
\end{split}
\end{align}

\bibliographystyle{IEEEtran}
\bibliography{Bib}

\end{document}